\documentclass[reqno,12pt,letterpaper]{amsart}
\usepackage{amsmath,amssymb,amsthm,graphicx,mathrsfs,url}
\usepackage[usenames,dvipsnames]{color}
\usepackage[colorlinks=true,linkcolor=Red,citecolor=Green]{hyperref}
\usepackage{amsxtra}

\def\?[#1]{\textbf{[#1]}\marginpar{\Large{\textbf{??}}}}

\let\epsilon=\varepsilon 

\newcommand{\C}{{\mathcal C}}
\newcommand{\RR}{{\mathbb R}}

\newcommand{\NN}{{\mathbb N}}
\newcommand{\CC}{{\mathbb C}}
\newcommand{\CI}{{{\mathcal C}^\infty}}
\newcommand{\CIc}{{{\mathcal C}^\infty_{\rm{c}}}}

\newtheorem{thm}{Theorem}
\newtheorem{prop}{Proposition}
\newtheorem{defi}[prop]{Definition}

\newtheorem{lem}{Lemma}

\numberwithin{equation}{section}

\DeclareMathOperator{\Det}{det}

\let\Im=\Imag

\let\Re=\Real

\DeclareMathOperator{\tr}{tr}

\title{Resonances as viscosity limits for exterior dilation analytic potentials}
\author{Haoren Xiong}
\email{xiong@math.berkeley.edu}
\address{Department of Mathematics, University of California,
Berkeley, CA 94720, USA}

\begin{document}

\maketitle

\begin{abstract}
For exterior dilation analytic potential, $V$, we use the method of complex scaling to show that the resonances of $ - \Delta + V $, in a conic neighbourhood of the real axis, are limits of eigenvalues of $ - \Delta + V - i \epsilon x^2 $ as
$ \epsilon \to 0+ $, if $V$ can be analytically extended from $\RR^n$ to a truncated cone in $\CC^n$.
\end{abstract}

\section{Introduction and statement of results}
\label{introduction}

We extend the results of \cite{Zw-vis}, when $V\in L_{\textrm{comp}}^\infty$, to the case of exterior dilation analytic potentials. For motivation and pointers to the literature we refer to \cite{Zw-vis}.

Thus, we consider
\[  H := - \Delta + V , \]
where $V$ is a real-valued potential which can be analytically extended from $\{x\in\RR^n\,:\,|x|> R\}$, for some $R>0$, to a truncated cone
\[\C_{\beta_0}^R := \{ z \in \CC^n\,:\,|\Im z| < \tan\beta_0 |\Re z|\ \textrm{and}\ |\Re z| > R \},\quad \beta_0 \leq \pi/8.\]
We still denote the analytic extension by $V$ and assume that
\begin{equation}
\label{eqn:decayV}
    \lim_{\C_{\beta_0}^R\owns |z|\to\infty}V(z) = 0.
\end{equation}
The resonances of $H$ are defined by the Aguiliar-Balslev-Combes-Simon theory, see \cite[\S 16, \S 18]{hislop2012}, \cite[\S 4.5]{res} and a review in \S\ref{Resonance}.

We now introduce a {\em regularized} operator,
\begin{equation}
\label{eqn:Heps}
 H_\epsilon := - \Delta - i \epsilon x^2 + V, \ \ \epsilon > 0 .
\end{equation}
(We write $ x^2 := x_1^2 + \cdots + x_n^2 $.) It is easy to see, with details reviewed in \S \ref{Eigenvalues}, that $ H_\epsilon $
is a non-normal unbounded operator on $ L^2 ( \RR^n ) $ with a discrete spectrum.
We have

\begin{thm}
\label{t:1}
Suppose that $ \{ z_j ( \epsilon) \}_{j=1}^\infty $ are
the eigenvalues of $ H_\epsilon $. Then, uniformly on any compact subsets of $ \{ z\,:\,-2\beta_0 < \arg z < 3\pi/2 + 2\beta_0 \} $,
\[   z_j ( \epsilon ) \to z_j ,  \ \ \epsilon \to 0 + , \]
where $ z_j $ are the resonances of $ H $.
\end{thm}

\medskip
\noindent
{\bf Notation.}
We use the following notation: $ f =  \mathcal O_\ell (
 g   )_H $ means that
$ \|f \|_H  \leq C_\ell  g $ where the norm (or any seminorm) is in the
space $ H$, and the constant $ C_\ell  $ depends on $ \ell $. When either
$ \ell $ or
$ H $ are absent then the constant is universal or the estimate is
scalar, respectively. When $ G = \mathcal O_\ell ( g ) : {H_1\to H_2 } $ then
the operator $ G : H_1  \to H_2 $ has its norm bounded by $ C_\ell g $.
Also when no confusion is likely to result, we denote the operator
$ f \mapsto g f $ where $ g $ is a function by $ g $.

\medskip
\noindent
{\sc Acknowledgments.} The author would like to thank
Maciej Zworski for helpful discussions. I am also grateful to the anonymous referee for the careful reading
of the first version and for the valuable comments.
This project was supported in part by
the National Science Foundation grant 1500852.

\section{spectral deformation and analytic vectors}
\label{SpecDeform}

We will review several basic concepts in the Aguilar-Balslev-Combes-Simon theory, such as spectral deformation and analytic vectors. For a detailed introduction, we refer to \cite[\S 17]{hislop2012} and the references given there.

Let $h \in \CI(\RR)$ be a non-decreasing function which satisfies
\begin{equation}
\label{eqn:h}
\begin{cases}
    h(t) = 0, & t<2R, \\
    h(t) = 1, & t>8R.
\end{cases}
\end{equation}
Moreover, we assume that
\begin{equation}
\label{eqn:chi}
\sup_{t\in\RR} h(t) + t h'(t) \leq 3/2.
\end{equation}
We define $g:\RR^n \to \RR^n$ as a smooth mapping by
\begin{equation}
\label{eqn:gmap}
g(x):=h(|x|)\,x = \begin{cases}
    0, & |x|<2R,  \\
    x, & |x|>8R,
\end{cases}
\end{equation}
and consider, for $\theta \in \RR$, the related family of maps $\phi_\theta:\RR^n \to \RR^n$ defined by
\begin{equation}
\label{eqn:phimap}
\phi_\theta(x) = x + \theta g(x)
\end{equation}
We let $Df$ denote the derivative of a map $f:\RR^n \to \RR^n$, then
\[ Dg(x) = h(|x|) I + |x|^{-1} h'(|x|) x\cdot x^T. \]
Using diagonalization, It is easy to see that
\begin{equation}
\label{eqn:diagonalDg}
0\leq h(|x|)I \leq Dg(x) \leq ( h(|x|) + |x|h'(|x|) ) I \leq 3/2\, I
\end{equation}
where $A \leq B$ means $B-A$ is positive semi-definite and the last inequality is implied by \eqref{eqn:chi}.
Hence $\sup_{x\in\RR^n} \|Dg(x)\| \leq 3/2$, where $\|\cdot\|$ denotes the operator norm on the set of linear transformation on $\RR^n$.
We note that $D\phi_\theta(x) = I + \theta(Dg)(x)$, if $|\theta| < 2/3$, then $D\phi_\theta$ is invertible by a Neumann series argument,
\[(D\phi_\theta)^{-1} = \sum_{j=0}^\infty (-1)^j \theta^j (Dg)^j.\]
Hence $\phi_\theta$ is a diffeomorphism of $\RR^n$ for $|\theta|< 2/3$ by the inverse function theorem.

We should remark that all the above argument is valid when we extend the definition \eqref{eqn:phimap} of $\phi_\theta$ to $\theta\in\CC$. We have $\phi_\theta:\RR^n \to \CC^n$ is a diffeomorphism provided $|\theta|< 2/3$.

We now introduce the behavior of functions under the action of the maps $\phi_\theta$. We first define $U_\theta$ for $\theta\in\RR$ by
\begin{equation}
\label{eqn:Utheta}
( U_\theta f )(x) = J_\theta(x)^{1/2} f(\phi_\theta(x))
\end{equation}
where $J_\theta (x)$ is the Jacobian of $\phi_\theta$,
\begin{equation}
\label{eqn:Jtheta}
J_\theta(x)= \Det D\phi_\theta(x) =\Det ( I + \theta(Dg)(x) ).
\end{equation}
It is east to see that $U_\theta,\ \theta\in\RR$ is unitary on $L^2(\RR^n)$ with the inverse $U_\theta^{-1}$ given by
\begin{equation}
\label{eqn:UthetaInverse}
( U_\theta^{-1} f )(x) = J_\theta ( \phi_\theta^{-1}(x) )^{-1/2} f ( \phi_\theta^{-1} (x) ).
\end{equation}
\eqref{eqn:diagonalDg} and \eqref{eqn:Jtheta} show that $J_\theta(x)^{1/2}$ extends analytically to complex $\theta$ provided $\theta< 2/3$. Hence, to extend the operators $U_\theta$ from $\theta\in\RR$ to $\theta\in\CC$, at least for small $|\theta|$, we need to find a dense set of functions $f$ in $L^2(\RR^n)$ that can be analytically extended on a small complex neighborhood of $\RR^n$ in $\CC^n$ such that $f\circ\phi_\theta \in L^2(\RR^n)$. For that we introduce the set of analytic vectors in $L^2(\RR^n)$.

\begin{defi}
\label{defn:analyticVec}
Let $\mathcal{A}$ be the linear space of all entire functions $f(z)$ having the property that in any conical region $\C_\epsilon$,
\[ \C_\epsilon:=\{ z\in\CC^n \,:\,|\Im z| \leq (1-\epsilon)\Re z \}, \]
for any $\epsilon>0$, we have for any $k\in\NN$,
\[ \lim_{z \in \C_\epsilon \to \infty} |z|^k |f(z)| = 0. \]
The set of analytic vectors in $L^2(\RR^n)$ are the restrictions to $\RR^n$ of $\mathcal A$, which is also denoted by $\mathcal A$.
\end{defi}

We define a domain $D_{\beta_0}$ in $\CC$ by
\begin{equation}
\label{eqn:Dbeta0}
D_{\beta_0} = \{ \theta\in\CC\,:\, |\Re\theta| + |\Im\theta| < \tan\beta_0 \}.
\end{equation}
Note that $D_{\beta_0}\subset \{ z\in\CC\,:\,|z|<1/2 \}$ since $\beta_0\leq \pi/8$, \eqref{eqn:diagonalDg} and \eqref{eqn:Jtheta} guarantee that the Jacobian $J_\theta$ is uniformly bounded for $\theta\in D_{\beta_0}$. Then, we recall the following results in \cite[Proposition 17.10]{hislop2012}:
\begin{prop}
\label{prop:specdeform}
Let $\mathcal{U}\equiv \{ U_\theta\,:\,\theta\in D_{\beta_0} \}$ be a spectral deformation family associated with vector field $g$ defined by \eqref{eqn:gmap}. Then,
\begin{itemize}
    \item the map $ (\theta,f)\in D_{\beta_0} \times \mathcal A \to U_\theta f $ is an $L^2$-analytic map ;
    \item for any $ \theta\in D_{\beta_0} $, $ U_\theta \mathcal A $ is dense in $L^2(\RR^n)$.
\end{itemize}
\end{prop}

We conclude this section with some properties about the deformation of $\RR^n\subset \CC^n$ under the map $\phi_\theta$ provided $\theta\in D_{\beta_0}$. We recall that $ \phi_\theta : \RR^n \to \CC^n $ is injective with the Jacobian $J_\theta \neq 0$ provided $|\theta|<2/3$. Hence $\phi_\theta(\RR^n)\subset \CC^n$ is an $n$-dimensional totally real submanifolds, see \cite[\S 4.5]{res}. Let $\Gamma_{a(\theta)}=\phi_\theta(\RR^n)$, where $a(\theta)\in (-\pi/2,\pi/2)$ is defined by
\begin{equation}
\label{eqn:tiltheta}
a(\theta) = \arg ( 1 + \theta ).
\end{equation}
In the literature about complex scaling, one can define $L^2(\Gamma_{a(\theta)})$ with volume element $|dw|=|J_\theta(x)|\,dx$ where $w=\phi_\theta(x)$ are the coordinates on $\Gamma_{a(\theta)}$, see \cite[\S 2.7, \S 4.5]{res} for details. Then we have the following:

\begin{prop}
\label{rem:Imagephitheta}
For any $\theta\in D_{\beta_0}$, $\Gamma_{a(\theta)}$ satisfies
\begin{gather}
\label{eqn:Gammatheta}
\begin{gathered}
\Gamma_{a(\theta)} \cap B_{\CC^n}(0,2R) = B_{\RR^n}(0,2R),  \\
\Gamma_{a(\theta)} \cap \CC^n \setminus B_{\CC^n}(0,12R) = e^{i a(\theta)}\RR^n \cap \CC^n \setminus B_{\CC^n}(0,12R),  \\
\Gamma_{a(\theta)} \subset \RR^n\cup \C_{\beta_0}^R.
\end{gathered}
\end{gather}
Furthermore, the spectral deformation operator $U_\theta$ extends to an isometry:
\[ U_\theta : L^2(\Gamma_{a(\theta)}) \to L^2(\RR^n). \]
\end{prop}

\begin{proof}
In view of \eqref{eqn:gmap} and \eqref{eqn:phimap}, it is easy to see that $\Gamma_{a(\theta)}=\phi_\theta(\RR^n)$ satisfies the first two equations of \eqref{eqn:Gammatheta}. For $\theta\in D_{\beta_0}$, we have
\[\frac{|\Im \phi_\theta(x)|}{|\Re \phi_\theta(x)|} = \frac{|\Im\theta| |\chi(|x|)|}{|1 + \Re\theta \chi(|x|)|} \leq \frac{|\Im\theta|}{1-|\Re\theta|} < \tan\beta_0,\]
where the last inequality is implied by \eqref{eqn:Dbeta0}. Moreover, $\phi_\theta(x)=x$ for $|x|<2R$, and $|\Re\phi_\theta(x)|\geq (1-|\Re\theta|)|x| > (1-\tan\beta_0)|x| > |x|/2 \geq R$ provided $|x|\geq 2R$, since $\beta_0\leq \pi/8$. Hence $\Gamma_{a(\theta)} \subset \RR^n\cup \C_{\beta_0}^R$.

Now we assume that $\theta\in D_{\beta_0}$, for any $f\in L^2(\Gamma_{a(\theta)})$, we can define $U_\theta f$ on $\RR^n$ by \eqref{eqn:Utheta}. To see $U_\theta f \in L^2(\RR^n)$, we compute directly:
\begin{equation}
\label{eqn:Uthetaisometry}
\begin{split}
    \int_{\RR^n} |U_\theta f(x)|^2 \, dx
    & = \int_{\RR^n} |J_\theta(x)^{1/2} f(\phi_\theta(x))|^2\,dx  \\
    & = \int_{\RR^n} |f(\phi_\theta(x))|^2 |J_\theta(x)|\,dx \\
    & = \int_{\Gamma_{a(\theta)}} |f(w)|^2 \,|dw| < \infty ,
\end{split}
\end{equation}
which also shows that $\|U_\theta f\|_{L^2(\RR^n)} = \|f\|_{L^2(\Gamma_{a(\theta)})}$ and $U_\theta$ is one-to-one. It remains to show that $U_\theta$ is onto. For $g\in L^2(\RR^n)$, let $G(w)=J_\theta(\phi_\theta^{-1}(w))^{-1/2}g(\phi_\theta^{-1}(w))$, $w\in \Gamma_{a(\theta)}$. We can follow \eqref{eqn:Uthetaisometry} to derive
\[ \int_{\Gamma_{a(\theta)}} |G(w)|^2 \,|dw| = \int_{\RR^n} |g(x)|^2 \,dx.\]
Hence $G\in L^2(\Gamma_{a(\theta)})$, then we conclude that $U_\theta$ is onto since $U_\theta G = g$.
\end{proof}

\section{resonances}
\label{Resonance}

We will follow Aguilar-Balslev-Combes-Simon theory to define the resonances of $H\equiv-\Delta+V$, see \cite[\S 16, \S 18]{hislop2012} and those resonances in a conic neighborhood of the real axis can be identified with the eigenvalues of certain non-self-adjoint operators associated with $H$. Using the analytic vectors $\mathcal A$, we recall the definition:

\begin{defi}
\label{defi:resonances}
The resonances of $H$ associated with analytic vectors $\mathcal{A}$ are the poles of the meromorphic continuations of all matrix elements $\langle f,R_H(z)g\rangle$ ($R_H(z)$ denotes the resolvent of $H$), $f,g\in\mathcal{A}$, from $\{z\in\CC\,:\,\Im z > 0\}$ to $\{z\in\CC\,:\,\Im z \leq 0\}$.
\end{defi}

First we introduce the spectral deformed Schr\"odinger operators $H(\theta)$ of $H$ associated with the spectral deformation family $\mathcal U = \{U_\theta\,:\,\theta\in D_{\beta_0} \}$. Consider, for $\theta\in D_{\beta_0}\cap\RR$,
\begin{equation}
\label{eqn:Htheta}
H(\theta): = U_\theta H U_\theta^{-1} = p_\theta^2 + V(\phi_\theta(x)),
\end{equation}
where
\begin{equation}
\label{eqn:ptheta}
p_\theta^2 = U_\theta p^2 U_\theta^{-1},\ \ p_j\equiv \frac{1}{i} \frac{\partial}{\partial x_j}.
\end{equation}
In view of Proposition \ref{rem:Imagephitheta}, we can extend $H(\theta)$ to $\theta\in D_{\beta_0}$. We recall the following basic facts about $p_\theta^2,\ \theta\in D_{\beta_0}$ in \cite[\S 18]{hislop2012}:
\begin{prop}
\label{prop:ptheta}
Let $p_\theta^2$ be as defined in \eqref{eqn:ptheta}, then $p_\theta^2,\ \theta\in D_{\beta_0}$ is an analytic family of operators with domain $D(p_\theta^2)=H^2(\RR^n)$. For the spectrum, we have $\sigma(p_\theta^2) = \sigma_{\textrm{ess}}(p_\theta^2) = e^{-2i a(\theta)}[0,\infty)$.
\end{prop}
\noindent
And for the resolvent $R_\theta(z):=(p_\theta^2 - z)^{-1}$ we have:
\begin{prop}
\label{prop:normRtheta}
For $\delta>0$, we have
\begin{equation}
\label{eqn:normRtheta}
R_\theta(z) = \mathcal O_\delta (1/|z|) : L^2(\RR^n)\to L^2(\RR^n) ,\ \ -2a(\theta)+\delta < \arg z < 2\pi-2a(\theta)-\delta.
\end{equation}
\end{prop}

\begin{proof}
We note that in the notation of Proposition \ref{rem:Imagephitheta},
\begin{equation}
\label{eqn:ptheta2}
    p_\theta^2 = U_\theta (-\Delta_{a(\theta)}) U_\theta^{-1},
\end{equation}
where $-\Delta_{a(\theta)}: H^2(\Gamma_{a(\theta)})\to L^2(\Gamma_{a(\theta)})$ is defined as the restriction of $\Delta_z$ to the totally real submanifold $\Gamma_{a(\theta)}$, see \cite[\S 4.5]{res}. Since $U_\theta:L^2(\Gamma_{a(\theta)}) \to L^2(\RR^n)$ and $U_\theta^{-1}:L^2(\RR^n) \to L^2(\Gamma_{a(\theta)})$ are both isometries, we have
\[ \|( p_\theta^2 - z )^{-1}\|_{L^2(\RR^n)\to L^2(\RR^n)} = \|(-\Delta_{a(\theta)}-z)^{-1}\|_{L^2(\Gamma_{a(\theta)})\to L^2(\Gamma_{a(\theta)})},\ \ z\notin e^{-2ia(\theta)}[0,\infty), \]
and thus \eqref{eqn:normRtheta} is a direct consequence of \cite[Theorem 4.35]{res}.
\end{proof}

Then we introduce some preliminary properties of the spectrum of $H(\theta)$:
\begin{prop}
\label{prop:spectrum}
There exists $R>0$ such that for any $\theta\in D_{\beta_0}$, we have
\[\sigma(H(\theta)) \cap i(R,\infty) = \emptyset.\]
As for the essential spectrum $\sigma_{\textrm{ess}}(H(\theta))$, we have more precisely,
\[\sigma_{\textrm{ess}}(H(\theta))=e^{-2ia(\theta)}[0,\infty).\]
\end{prop}

\noindent
{\bf Remark:} In fact, $\sigma(H(\theta))\cap\{z:0<\arg z<2\pi-2a(\theta)\}$ is discrete and lies in $(-\infty,0)$, which is a consequence of the following Lemma \ref{lem:resonance}.
\begin{proof}
For $\theta\in D_{\beta_0}$, we have
\begin{equation}
\label{eqn:Hresolvent}
( H(\theta) - z )R_\theta(z) = I + V(\phi_\theta(x))R_\theta(z).
\end{equation}
Now assume $z\in i(R,\infty)$, note that $\theta\in D_{\beta_0} \implies -\beta_0 < a(\theta) < \beta_0$, we have
\[-2a(\theta) + \pi/4 < \arg z < 2\pi - 2a(\theta) - \pi/4,\ \ \textrm{for all }\theta\in D_{\beta_0}.\]
Using \eqref{eqn:normRtheta}, we see that $R_\theta(z) = \mathcal{O}(1/|z|):L^2(\RR^n)\to L^2(\RR^n)$ for all $z\in i(R,\infty)$ and $\theta\in D_{\beta_0}$. Recalling $\phi_\theta(\RR^n)\subset \RR^n \cup C_{\beta_0}^R$ and $V\in L^\infty(\RR^n \cup C_{\beta_0}^R)$,
we conclude that
\[ \sup_{z\in i(R,\infty)} \| V(\phi_\theta(x)) R_\theta(z) \|_{L^2(\RR^n)\to L^2(\RR^n)} = \mathcal{O}(R^{-1}),\quad\textrm{uniformly for }\theta\in D_{\beta_0}. \]
Then for $R\gg 1$, $I + V(\phi_\theta(x))R_\theta(z)$ is invertible using the Neumann series:
\[ ( I + V(\phi_\theta(x))R_\theta(z) )^{-1} = \sum_{j=0}^\infty ( V(\phi_\theta(x))R_\theta(z) )^j.
\]
Hence $H(\theta)-z$ is invertible by (\ref{eqn:Hresolvent}), for all $z\in i(R,\infty)$.

For the essential spectrum $\sigma_{\textrm{ess}}(H(\theta))$, note that $\sigma_{\textrm{ess}}(p_\theta^2) = e^{-2ia(\theta)}[0,\infty)$ in Proposition \ref{prop:ptheta}, by the invariance under compact perturbations, it suffices to show that $V(\phi_\theta(x))$ is $p_\theta^2$-compact, i.e. $V(\phi_\theta(x)):D(p_\theta^2) = H^2(\RR^n) \to L^2(\RR^n)$ is compact. Since $H^2(B_{\RR^n}(0,R))\Subset L^2(B_{\RR^n}(0,R)),\ \forall\,R>0$, and $V\circ\phi_\theta \in L^\infty(\RR^n),\ \ V(\phi_\theta(x))\to 0,\ x\to\infty$ by \eqref{eqn:decayV}, it is easy to see the compactness of $V(\phi_\theta(x)):\ H^2(\RR^n)\to L^2(\RR^n)$.
\end{proof}

Now we state the main result in this section, in which we identify the resonances defined in Definition \ref{defi:resonances} as the eigenvalues of certain spectral deformed operators $H(\theta)$.
\begin{lem}
\label{lem:resonance}
Let $H = -\Delta + V$ be a self-adjoint Schr\"odinger operator with a real-valued potential $V$ satisfying our assumptions as in \S \ref{introduction}. Then for any $\theta\in D_{\beta_0} \cap \CC^+$, we have:
\begin{itemize}
    \item For $f,g\in\mathcal{A}$, the function
        \begin{equation}
        \label{eqn:F_fg}
        F_{f,g}(z)\equiv\langle f,R_H(z)g\rangle,
        \end{equation}
        defined for $\Im z > 0$, has a meromorphic continuation across $[0,\infty)$ into $S_\theta^- \equiv \CC\setminus e^{-2ia(\theta)}[0,\infty)$.
    \item The poles of the meromorphic continuations of all matrix elements $F_{f,g}(z)$ into $S_\theta^-$ are eigenvalues of the operator $H(\theta)$.
\end{itemize}
\end{lem}

\begin{proof}
With $F_{f,g}(z)$ defined in (\ref{eqn:F_fg}), the assumption on $V$ implies that $F_{f,g}$ is analytic on $\CC^+ \equiv \{z\in\CC\,:\,\Im z > 0\}$. Fix $z\in\CC^+$. For $\theta\in D_{\beta_0} \cap \RR$, $U_\theta$ is unitary and thus we can write
\begin{equation}
\label{eqn:unitaryFfg}
F_{f,g}(z) = \langle U_\theta f , (U_\theta R_H(z) U_\theta^{-1})U_\theta g \rangle = \langle U_\theta f, R_{H(\theta)}(z) U_\theta g \rangle.
\end{equation}
Proposition \ref{prop:spectrum} implies that $\theta\in D_{\beta_0} \to R_{H(\theta)}(z)$ is an analytic map provided $z\in i(R,\infty)$. Since we can write $U_{\bar{\theta}}f$ instead of $U_\theta f$ in (\ref{eqn:unitaryFfg}), we have
\begin{equation}
\label{eqn:meroextensionFfg}
\theta\in D_{\beta_0}\to F_{f,g}(z;\theta)\equiv \langle U_{\bar{\theta}}f, R_{H(\theta)}(z)U_\theta g \rangle
\end{equation}
is an analytic map provided $z\in i(R,\infty)$. Hence for any $z\in i(R,\infty)$, we have
\[F_{f,g}(z;\theta)=F_{f,g}(z),\quad \forall\,\theta\in D_{\beta_0}, \]
since this is true for all $\theta\in D_{\beta_0}\cap \RR$. Now fix any $\theta\in D_{\beta_0}\cap \CC^+$, Proposition \ref{prop:spectrum} guarantees that $F_{f,g}(z;\theta)$ can be meromorphically continued from $i(R,\infty)$ to $S_\theta^-$ since $\sigma_{\textrm{ess}}(H(\theta)) \cap S_\theta^- = \emptyset$. We have shown that $F_{f,g}(z;\theta) = F_{f,g}(z),\ z\in i(R,\infty)$, then by the identity principle for meromorphic functions, we conclude that $F_{f,g}(z;\theta)$ is a meromorphic continuation of $F_{f,g}(z)$ from $\CC^+$ to $S_\theta^-$.

Recalling that $F_{f,g}(z;\theta)=\langle U_{\bar{\theta}}f, (H(\theta)-z)^{-1} U_{\theta}g \rangle$ and that $U_{\theta}\mathcal{A}$, $U_{\bar{\theta}}\mathcal{A}$ are both dense in $L^2(\RR^n)$, thus if $H(\theta)$ has an eigenvalue at $\lambda_\theta \in S_\theta^-$, there must exist $f,g\in\mathcal{A}$ such that $\lambda_\theta$ is a pole of $F_{f,g}(z;\theta)$. Conversely, if $F_{f,g}(z;\theta)$ has a pole $\lambda_\theta\in S_\theta^-$, then it must be an eigenvalue of $H(\theta)$.
\end{proof}

\noindent
{\bf Remark:} For nonzero resonance $\lambda$ of $H$, we can define its multiplicity as the (algebraic) multiplicity of $\lambda$ as an eigenvalue of some $H(\theta)$. More precisely, let $\lambda\in \{z : -2\beta_0 < \arg z < 3\pi/2 + 2\beta_0\}$ be a resonance of $H$, there exists $\theta\in D_{\beta_0}\cap \CC^+$ such that $-2a(\theta) < \arg \lambda$. Lemma \ref{lem:resonance} implies that $\lambda$ is also an eigenvalue of $H(\theta)$, then we define the multiplicity of resonance $\lambda$ as follows:
\begin{equation}
\label{eqn:resmult}
m(\lambda) := m_\theta(\lambda) \equiv -\frac{1}{2\pi i}\tr\oint_\lambda (H(\theta) - z)^{-1} dz,
\end{equation}
where the integral is over a positively oriented circle enclosing $\lambda$ and containing no eigenvalues of $H(\theta)$ other than $\lambda$. To see that the multiplicity $m(\lambda)$ is well-defined, we need to show that $m(\lambda)$ does not depend on the choice of $\theta$. Assume $\theta_0,\theta_1 \in D_{\beta_0}$ satisfy $-2a(\theta_0)\leq -2a(\theta_1) < \arg \lambda$, let $\theta_t = (1-t)\theta_0 + t\theta_1$ then $-2a(\theta_t)<\arg \lambda$ for all $t\in [0,1]$. Let $C_\lambda$ be a positively oriented circle enclosing $\lambda$ with sufficiently small radius such that $C_\lambda\subset \{z: \arg z > -2a(\theta_1)\}$ and contains no resonances of $H$ other than $\lambda$. Therefore, $C_\lambda$ contains no eigenvalues of $H(\theta_t)$ other than $\lambda$ for all $t\in [0,1]$ as a consequence of Lemma \ref{lem:resonance}. Now we have
\[m_{\theta_t}(\lambda) = -\frac{1}{2\pi i}\tr\int_{C_\lambda} (H(\theta_t)-z)^{-1}dz,\quad t\in[0,1].\]
Hence $m_{\theta_t}(\lambda)$ depends continuously on $t$ which implies that $m_{\theta_t}(\lambda)$ must be a constant as it is integer-valued. In particular, we have $m_{\theta_0}(\lambda) = m_{\theta_1}(\lambda)$, thus $m(\lambda)$ is well-defined.

\section{Eigenvalues and complex scaling}
\label{Eigenvalues}

In this section we will show that the eigenvalues of $H_\epsilon \equiv -\Delta - i\epsilon x^2 + V$ are invariant under complex scaling, in other words, these eigenvalues are the same as the eigenvalues of
\[
H_\epsilon(\theta) := U_\theta H_\epsilon U_\theta^{-1} = p_\theta^2 - i\epsilon \phi_\theta(x)^2 + V(\phi_\theta(x)),\quad \theta\in D_{\beta_0}\cap \CC^+.
\]

First we recall some basic properties about the Davies harmonic oscillator and its deformation, see \cite[\S 3]{Zw-vis} for details. The operator $ H_{\epsilon, \gamma} := - \Delta + e^{ - i \gamma } \epsilon x^2 $,
$ \epsilon > 0 $, $0 \leq \gamma < \pi $, was used by Davies \cite{Dav} to illustrate properties of
non-normal differential operators. We are more interested in the deformations of $H_{\epsilon,\gamma}$ under complex scaling. Let
\[ Q_{\epsilon,\theta}= -\Delta_\theta -i\epsilon x_\theta^2,\ \ \textrm{where }x_\theta=z|_{\Gamma_\theta} \]
be a deformed operator on $\Gamma_\theta$ as in \cite[\S 3]{Zw-vis}. In view of \eqref{eqn:ptheta2}, we have
\begin{equation}
\label{eqn:Qepstheta}
p_\theta^2 - i\epsilon \phi_\theta(x)^2 = U_\theta Q_{\epsilon,a(\theta)} U_\theta^{-1},\ \ \theta\in D_{\beta_0}.
\end{equation}
Hence we can study the spectrum and the resolvents of $p_\theta^2 - i\epsilon \phi_\theta(x)^2$ using the relevant results about $Q_{\epsilon,a(\theta)}$. We recall \cite[Lemma 4.]{Zw-vis} that $\sigma(Q_{\epsilon,a(\theta)})=\sqrt{\epsilon} e^{-i\pi/4}(n + 2|\NN_0^n|)$ for $\theta\in D_{\beta_0}$, then by \eqref{eqn:Qepstheta} we have
\begin{prop}
\label{prop:QepsSpectrum}
For $\theta\in D_{\beta_0}$, $\epsilon>0$, the spectrum of $p_\theta^2 - i\epsilon \phi_\theta(x)^2$ is independent of $\theta$ and given by $\sqrt{\epsilon} e^{-i\pi/4}(n + 2|\NN_0^n|)$.
\end{prop}

\noindent
For the resolvents of $p_\theta^2 - i\epsilon \phi_\theta(x)^2$:
\begin{equation}
\label{eqn:Repstheta}
R_{\epsilon, \theta}(z) := ( p_\theta^2 - i\epsilon \phi_\theta(x)^2 - z )^{-1},\ \ \theta\in D_{\beta_0},
\end{equation}
we recall \cite[Lemma 5.]{Zw-vis} that for $\delta>0$, $-\pi/8<\theta<\pi/8$, we have
\[ (Q_{\epsilon,\theta}-z)^{-1}=\mathcal{O}_\delta(1/|z|):L^2(\Gamma_\theta)\to L^2(\Gamma_\theta),\ \ -2\theta+\delta < \arg z < 3\pi/2 + 2\theta - \delta, \]
uniformly for $0<\epsilon<\epsilon_0$, where $\epsilon_0 > 0$ is a constant. Using \eqref{eqn:Qepstheta}, we have
\begin{prop}
\label{prop:normRepstheta}
Let $\theta\in D_{\beta_0}$, $\delta>0$, then uniformly for $0<\epsilon<\epsilon_0$, we have
\begin{equation}
\label{eqn:normRepstheta}
R_{\epsilon,\theta}(z)=\mathcal{O}_\delta(1/|z|):L^2(\RR^n)\to L^2(\RR^n),\ \ -2a(\theta)+\delta < \arg z < 3\pi/2 + 2a(\theta) - \delta.
\end{equation}
\end{prop}

Now we state the main result about the eigenvalues of $H_\epsilon$ :
\begin{lem}
\label{lem:eigenvalue}
For any $\theta\in D_{\beta_0}$, $0 < \epsilon < \epsilon_0$,
\[z \mapsto R_{H_\epsilon(\theta)}(z)\equiv ( H_\epsilon(\theta) - z )^{-1},\quad -\pi/4 < \arg z < 7\pi/4,
\]
is a meromorphic family of operators on $L^2(\RR^n)$ with poles of finite rank. Furthermore, the poles of $( H_\epsilon(\theta) - z )^{-1}$ do not depend on $\theta\in D_{\beta_0}$ and coincide, with agreement of multiplicities, with the poles of $( H_\epsilon - z )^{-1}$.
\end{lem}

\begin{proof}
For fixed $\theta\in D_{\beta_0}$, one can compute
\[
( H_\epsilon(\theta) - z )R_{\epsilon,\theta}(z) = I + V(\phi_\theta(x))R_{\epsilon,\theta}(z),
\]
then we obtain from \eqref{eqn:normRepstheta} that
\begin{gather}
\label{eqn:HepsilonResolvent}
\begin{gathered}
R_{H_\epsilon(\theta)}(z) = R_{\epsilon,\theta}(z) ( I + V(\phi_\theta(x) )R_{\epsilon,\theta}(z) )^{-1} , \\
-2a(\theta) + \delta  < \arg z < 3\pi/2 + 2a(\theta) - \delta , \ \  |z| \gg 1 ,
\end{gathered}
\end{gather}
where for large $|z|$, $ I + V(\phi_\theta(x) )R_{\epsilon,\theta}(z) $ is invertible by a Neumann series argument. Note that $R_{\epsilon,\theta}(z): L^2(\RR^n)\to H^2(\RR^n),\ \ \arg z\neq -\pi/4$ by Proposition \ref{prop:QepsSpectrum}, recalling that $V(\phi_\theta(x) ):H^2(\RR^n)\to L^2(\RR^n)$ is compact (see the proof of Proposition \ref{prop:spectrum}), we conclude that $z\mapsto V(\phi_\theta(x))R_{\epsilon,\theta}(z)$ is an analytic family of compact operators for $-\pi/4 < z <7\pi/4$. Hence $z\mapsto ( I + V(\phi_\theta(x))R_{\epsilon,\theta}(z) )^{-1}$ is a meromorphic family of operators in the same range of $z$. In particular $z\mapsto R_{H_\epsilon(\theta)}(z),\ \ -\pi/4 < z <7\pi/4$ is a meromorphic family of operators on $L^2(\RR^n)$ with poles of finite rank.

The poles and their multiplicities are independent of $\theta$. For that we modify the proof of Lemma \ref{lem:resonance} and define matrix elements:
\[G_{f,g}(z) = \langle f, ( H_\epsilon - z )^{-1} g \rangle,\]
and
\[G_{f,g}(z;\theta) = \langle U_{\bar{\theta}}f, (H_\epsilon(\theta)-z)^{-1} U_{\theta}g \rangle,\]
for all $f,g\in\mathcal{A}$.

Note that $-2a(\theta)+\pi/4 < \pi/2 < 3\pi/2 + 2a(\theta) -\pi/4$ since $-\beta_0<a(\theta)<\beta_0,\ \theta\in D_{\beta_0}$, using \eqref{eqn:normRepstheta} and Neumann series argument, $H_{\epsilon}(\theta) - z$ is invertible at $z=i\rho,\ \rho\gg 1$ for each $\theta\in D_{\beta_0}$. Like \eqref{eqn:meroextensionFfg}, we have
\[\theta\in D_{\beta_0}\to G_{f,g}(z;\theta)\equiv \langle U_{\bar{\theta}}f, R_{H(\theta)}(z)U_\theta g \rangle \]
is an analytic map provided $z=i\rho,\ \rho\gg 1$. Hence we have
\begin{equation}
\label{eqn:localGfg=Gfgtheta}
G_{f,g}(z;\theta) = G_{f,g}(z),\quad \forall\,\theta\in D_{\beta_0},\quad z=i\rho,\ \rho\gg 1,
\end{equation}
since this is true for all $\theta\in D_{\beta_0}\cap \RR$. Now fix any $\theta\in D_{\beta_0}$, note that $G_{f,g}(z)$ and $G_{f,g}(z;\theta)$ are both meromorphic in $-\pi/4 < z < 7\pi/4$, we conclude that
\begin{equation}
\label{eqn:Gfg=Gfgtheta}
G_{f,g}(z;\theta) = G_{f,g}(z),\quad -\pi/4 < z < 7\pi/4,
\end{equation}
by \eqref{eqn:localGfg=Gfgtheta} and the identity principle of meromorphic functions.

Now argue as in the end of the proof of Lemma \ref{lem:resonance}: if $(H_\epsilon-z)^{-1}$ has a pole at $\lambda_\theta\in \CC\setminus e^{-i\pi/4}[0,\infty)$, then there must exist $f,g\in\mathcal{A}$ such that $\lambda_\theta$ is a pole of $G_{f,g}(z;\theta)$, by \eqref{eqn:Gfg=Gfgtheta}, $\lambda_\theta$ is also a pole of $G_{f,g}(z)$ thus $(H_\epsilon(\theta)-z)^{-1}$ must have a pole at $\lambda_\theta$ and vise versa. Hence for any $\theta\in D_{\beta_0}$, the poles of $(H_\epsilon(\theta)-z)^{-1}$ in $\CC\setminus e^{-i\pi/4}[0,\infty)$ coincide the poles of $(H_\epsilon-z)^{-1}$ in $\CC\setminus e^{-i\pi/4}[0,\infty)$.

To show the agreement of multiplicities, for any pole $\lambda$ of $(H_\epsilon(\theta)-z)^{-1}$, the multiplicity of $\lambda$ is defined by
\[ m_{\epsilon,\theta}(\lambda) = -\frac{1}{2\pi i} \tr \oint_\lambda (H_\epsilon(\theta)-z)^{-1} dz,\]
where the integral is over a positively oriented circle independent of $\theta$ enclosing $\lambda$ and containing no poles other than $\lambda$.
Since $m_{\epsilon,\theta}(\lambda)$ is continuous on $\theta\in D_{\beta_0}$ and integer-valued, it must be independent of $\theta\in D_{\beta_0}$. Hence we have
\[ m_{\epsilon,\theta}(\lambda) = m_{\epsilon,0}(\lambda)= -\frac{1}{2\pi i} \tr \oint_\lambda (H_\epsilon-z)^{-1} dz \]
which is the multiplicity of $\lambda$ as a pole of $(H_\epsilon-z)^{-1}$.
\end{proof}

\section{Meromorphic continuation}
\label{mc}

In this section we will introduce a new way to express the meromorphic continuations of resolvents $R_{H(\theta)}(z)$ and $R_{H_\epsilon(\theta)}(z)$ in a given region $\Omega\Subset \{ z: - 2a(\theta) < \arg z < 3 \pi /2 + 2a(\theta) \},$ which is crucial in the proof of Theorem \ref{t:1}. For that we will first review some properties about $R_\theta(z)$ and the weighted $L^2$ space, $\langle x\rangle^{-2} L^2(\RR^n)$.

\begin{lem}
\label{lem:cpct embed}
Let $\langle x\rangle^{-2} L^2(\RR^n)$ be a weighted $L^2$ space with the norm
\begin{equation}
\label{eqn:weighted norm}
\|u\|_{\langle x\rangle^{-2} L^2(\RR^n)}=\|\langle x\rangle^2 u\|_{L^2(\RR^n)}.
\end{equation}
Then $H^2(\RR^n)\cap\langle x\rangle^{-2} L^2(\RR^n)$ is compactly embedded in $L^2(\RR^n)$.
\end{lem}

\begin{proof}
Let $u_n\in H^2(\RR^n)\cap\langle x\rangle^{-2} L^2(\RR^n)$ with $\|u_n\|_{H^2(\RR^n)}\leq 1$ and $\|\langle x\rangle^2 u_n\|_{L^2(\RR^n)}\leq 1$. For some $r>0$ to be decided, we have
\[\int_{|x|\geq r}|u_n(x)|^2 dx \leq \langle r\rangle^{-4}\int_{|x|\geq r} \langle x\rangle^4 |u_n(x)|^2 dx \leq \langle r\rangle^{-4}\|\langle x\rangle^2 u_n\|_{L^2(\RR^n)}^2 = \langle r\rangle^{-4}.  \]
Then we choose $r$ sufficientlt large such that $\int_{|x|\geq r}|u_n(x)|^2 dx < 1/8$ for all $n$. Since $H^2(B(0,r)\Subset L^2(B(0,r)$, there exists subsequence $\{u_n^{(1)}\}\subset\{u_n\}$ satisfying
\[\int_{B(0,r)}|u_n^{(1)}(x) - u_m^{(1)}(x)|^2 dx < 1/2,\quad\textrm{for all }n,m.\]
Hence we have
\[
\begin{split}
\|u_n^{(1)}-u_m^{(1)}\|_{L^2(\RR^n)}^2 & = \int_{B(0,r)}|u_n^{(1)}(x) - u_m^{(1)}(x)|^2 dx + \int_{|x|\geq r}|u_n^{(1)}(x) - u_m^{(1)}(x)|^2 dx \\
& < 1/2 + \int_{|x|\geq r}(2|u_n^{(1)}(x)|^2 + 2|u_m^{(1)}(x)|^2) dx \\
& < 1/2 + 2/8 +2/8 = 1.
\end{split}
\]
By the same argument, we can find $\{u_n^{(1)}\}\supset\cdots\supset\{u_n^{(j)}\}\supset\cdots$ with
\[\|u_n^{(j)}-u_m^{(j)}\|_{L^2(\RR^n)} < 1/j,\quad\textrm{for all }n,m.\]
Then the subsequence $\{u_j^{(j)}\}\subset\{u_n\}$ is a Cauchy sequence in $L^2(\RR^n)$.
\end{proof}

\begin{lem}
\label{lem:Rtheta}
Fix $\theta \in D_{\beta_0}\cap \CC^+$, $R_\theta(z)$ is an analytic family of operators $\langle x\rangle^{-2}L^2\to \langle x\rangle^{-2}L^2$ for $ -2a(\theta)< \arg z < 2 \pi  - 2a(\theta) $. Furthermore, if $\Omega\Subset \{ z: - 2a(\theta) < \arg z < 2 \pi - 2a(\theta) \}$ then there exists $C=C_{\Omega,\theta}$ such that
\[
\|R_\theta(z)\|_{\langle x\rangle^{-2}L^2\to \langle x\rangle^{-2}L^2} \leq C,\quad z\in\Omega.
\]
\end{lem}

\begin{proof}
In view of \eqref{eqn:h} and \eqref{eqn:Dbeta0}, we have
\[|x|/2 < |\phi_\theta(x)|=|1+\theta h(|x|)||x| < 3|x|/2\ \implies\ \langle x\rangle/2 < \langle \phi_\theta(x)\rangle < 3\langle x\rangle/2.\]
Then it is equivalent to prove the lemma with $\langle \phi_\theta(x)\rangle$ replacing $\langle x \rangle $. We recall Proposition \ref{rem:Imagephitheta} to write
\[
\langle \phi_\theta(x)\rangle^2 R_\theta(z) \langle \phi_\theta(x)\rangle^{-2} = U_{a(\theta)}\langle x_{a(\theta)} \rangle^2 ( -\Delta_{a(\theta)} - z )^{-1} \langle x_{a(\theta)} \rangle^{-2} U_{a(\theta)}^{-1},
\]
where $x_{a(\theta)}$ is the coordinate on $\Gamma_{a(\theta)}$. Then it suffices to show that, for any $0 < \alpha <\beta_0$,
\begin{equation}
\label{eqn:wGamma}
\langle w \rangle ^2 (-\Delta_\alpha - \lambda^2)^{-1} \langle w \rangle ^{-2} : L^2(\Gamma_\alpha)\to L^2(\Gamma_\alpha), \ \ \Im (e^{i\alpha}\lambda) > 0,
\end{equation}
is analytic with uniformly bounded norm provided $\lambda$ in any compact subset of $\{\lambda\in\CC:\ Im (e^{i\alpha}\lambda)\}$, where $w$ denotes the coordinate on $\Gamma_\alpha$. To prove \eqref{eqn:wGamma}, consider the integral kernel of that operator:
\begin{equation}
\label{eqn:Kintegralkernel}
    K(\lambda,w_1,w_2) = \langle w_1\rangle^2 R_0(\lambda,w_1,w_2) \langle w_2\rangle^{-2},\ \ w_1,w_2\in\Gamma_\alpha,
\end{equation}
where $R_0(\lambda,w_1,w_2)$ is the integral kernel of $ ( -\Delta_\alpha - \lambda^2 )^{-1}:\,L^2(\Gamma_\alpha) \to L^2(\Gamma_\alpha) $. It is easy to see that
\begin{equation}
\label{eqn:K_estimate_1}
\begin{split}
    |K(\lambda,w_1,w_2)| & \leq (1+|w_1|^2) \, |R_0(\lambda,w_1,w_2)| \, \langle w_2 \rangle^{-2}  \\
    & \leq 2( 1 + |w_1-w_2|^2 + |w_2|^2 ) \,( 1 + |w_2|^2 ) ^{-1} \, |R_0(\lambda,w_1,w_2)|  \\
    & \leq 2 ( 1 + |w_1-w_2|^2 ) |R_0(\lambda,w_1,w_2)|.
\end{split}
\end{equation}
To introduce the explicit formula of $R_0(\lambda,w_1,w_2)$, we recall that one can define $((w_1-w_2)\cdot(w_1-w_2))^{1/2}$ for $w_1,w_2\in\Gamma_\alpha$, see \cite[\S 4.5]{res}. Then we can write
\[ R_0(\lambda,w_1,w_2) = C_n \lambda^{n-2} (\lambda((w_1-w_2)\cdot(w_1-w_2))^{1/2})^{-\frac{n-2}{2}}H^{(1)}_{\frac{n}{2}-1}(\lambda((w_1-w_2)\cdot(w_1-w_2))^{1/2}) \]
where $H^{(1)}_k$ denote the Hankel functions of the first kind, and we can estimate $ |R_0(\lambda,w_1,w_2)| $ as follows:
\begin{equation}
\label{eqn:R0estimate}
|R_0(\lambda,w_1,w_2)| \leq \frac{P_n(\lambda((w_1-w_2)\cdot(w_1-w_2))^{1/2})}{(((w_1-w_2)\cdot(w_1-w_2))^{1/2})^{n-2}} e^{-\Im \lambda((w_1-w_2)\cdot(w_1-w_2))^{1/2}}
\end{equation}
where $P_n$ is a polynomial of degree $(n-3)/2$, see \cite[\S 2.2]{GS2014} and \cite[\S 4.5]{res} for details. Using \eqref{eqn:Gammatheta}, it is easy to see that for any $\delta$ small, there exists $C_\delta>0$ such that $ | \arg ((w_1-w_2)\cdot(w_1-w_2))^{1/2} - \alpha | < \delta $ provided $|w_1-w_2|>C_\delta$. Note that $0 < \arg\lambda + \alpha < \pi$, for every $\lambda$, we can choose $\delta=\delta_\lambda$ such that $2\delta < \arg\lambda + \alpha < \pi-2\delta$, then for $|z-w|>C_\lambda$, we have
\[\delta < \arg \lambda((w_1-w_2)\cdot(w_1-w_2))^{1/2} < \pi-\delta,\]
and thus
\begin{equation}
\label{eqn:expdecay}
 e^{-\Im \lambda((w_1-w_2)\cdot(w_1-w_2))^{1/2}} < e^{-c_\lambda |w_1-w_2|},\ c_\lambda > 0, \ \ \textrm{if }|w_1-w_2|> C_\lambda .
\end{equation}
Then using \eqref{eqn:K_estimate_1}, \eqref{eqn:R0estimate} and \eqref{eqn:expdecay}, we conclude that
\[\sup_{w_1\in\Gamma_\alpha}\int_{\Gamma_\alpha}|K(\lambda,w_1,w_2)| dw_1 < M_\lambda,\quad \sup_{w_2\in\Gamma_\alpha}\int_{\Gamma_\alpha}|K(\lambda,w_1,w_2)| dz < M_\lambda.\]
By Schur criterion, we proved \eqref{eqn:wGamma}, the analyticity in $\lambda$ is easy to see using the explicit formula of $R_0(\lambda,w_1,w_2)$. If $\lambda\in K \Subset \{\lambda\in\CC:\ Im (e^{i\alpha}\lambda)\}$, then there exist $c_K$ and $C_K$ such that
\[
e^{-\Im \lambda((w_1-w_2)\cdot(w_1-w_2))^{1/2}} < e^{-c_K |w_1-w_2|},\ c_K > 0, \ \ \textrm{if }|w_1-w_2|> C_K.
\]
Follow the above argument, there exists $M=M_K>0$ such that
\[
\|\langle w \rangle ^2 (-\Delta_\alpha - \lambda^2)^{-1} \langle w \rangle ^{-2}\|_{L^2(\Gamma_\alpha)\to L^2(\Gamma_\alpha)} < M_K,\quad \textrm{for all }\lambda \in K,
\]
which completes the proof.
\end{proof}

Now we state the main result of this section:
\begin{lem}
\label{lem:MeroContinuation}
Fix any $ \theta\in D_{\beta_0}\cap \CC^+ $ and $\Omega\Subset \{ z: - 2a(\theta) < \arg z < 3 \pi /2 + 2a(\theta) \}$
, there exists $\chi\in\CIc(\RR^n)$, $\chi\equiv 1$ on $B(0,T)$ for some $T>0$ such that for $0 \leq \epsilon <\epsilon_0$, $ H_\epsilon(\theta) - \chi V - z $ is invertible in $\Omega$ and
\[  z \mapsto  ( I + R_{H_\epsilon(\theta)-\chi V} ( z )\, \chi V )^{-1} ,\quad z\in \Omega,
\]
is a meromorphic family of operators on $ L^2 ( \RR^n ) $ with poles of finite rank, where we write $R_{H_\epsilon(\theta)-\chi V}(z)=( H_\epsilon(\theta) - \chi V - z )^{-1}$ for simplicity.
Moreover,
\begin{equation}
\label{eqn:Hieps}
m_{\epsilon,\theta} ( z ) := \frac{ 1}{ 2 \pi i} \tr
\oint_z   ( I + R_{H_\epsilon(\theta)-\chi V} ( w )\,\chi V )^{-1} \partial_w (R_{H_\epsilon(\theta)-\chi V} ( w )\,\chi V ) dw ,
\end{equation}
where the integral is over a positively oriented circle enclosing $ z $
and containing no poles other than possibly $ z $, satisfies \begin{equation}
\label{eqn:lmult}
m_{\epsilon,\theta} ( z ) =  \frac{1}{ 2 \pi i } \tr \oint_z
( w - H_\epsilon(\theta) )^{-1} dw,\quad 0 \leq \epsilon < \epsilon_0,
\end{equation}
where  $H_0(\theta) = H(\theta)$.
\end{lem}

\begin{proof}
We modify the argument in \cite[\S 4]{Zw-vis} to our setting. First there exists $\delta=\delta_\Omega$ such that $\Omega\subset \C_\delta := \{ z : -2a(\theta) + \delta < \arg z < 3\pi/2 + 2a(\theta) -\delta,\ |z|>\delta\}$, we recall \eqref{eqn:normRtheta} and \eqref{eqn:normRepstheta} that uniformly for $0 \leq \epsilon < \epsilon_0$, we have
\begin{equation}
\label{eqn:normRepsthetaz}
R_{\epsilon,\theta} ( z ) = \mathcal O_\delta ( 1/|z| ) : L^2 ( \RR^n ) \to L^2 ( \RR^n ) , \quad z \in \C_\delta.
\end{equation}
Hence $\|R_{\epsilon,\theta} ( z )\|_{ L^2 ( \RR^n ) \to L^2 ( \RR^n )} < C_\delta,\ \forall z\in\C_\delta$, for some $C_\delta>0$. In view of \eqref{eqn:decayV}, for $T$ sufficiently large, we have $\|(1-\chi)V\|_{L^\infty}<1/2C_\delta$ and thus
\begin{equation}
\label{eqn:Neumann estimate}
\|R_{\epsilon,\theta} ( z ) (1-\chi)V\|_{ L^2 ( \RR^n ) \to L^2 ( \RR^n )}<1/2,\quad \textrm{for all }z\in\C_\delta.
\end{equation}
Then $( I + R_{\epsilon,\theta} ( z ) (1-\chi)V )$ is invertible by the Neumann series argument, which implies that $ H_\epsilon(\theta) - \chi V - z $ is invertible and
\begin{equation}
\label{eqn:RH-chiV}
R_{H_\epsilon(\theta)-\chi V}(z)=( H_\epsilon(\theta) - \chi V - z )^{-1}=( I + R_{\epsilon,\theta} ( z ) (1-\chi)V )^{-1} R_{\epsilon,\theta} ( z ),\quad\forall z\in \C_\delta.
\end{equation}
Since $\chi V\in L^\infty(\RR^n)$, \eqref{eqn:normRepsthetaz} and \eqref{eqn:RH-chiV} imply that for $z\in\C_\delta$, $|z|\gg 1$, both $I + \chi V R_{H_\epsilon(\theta)-\chi V}(z)$ and $I + R_{H_\epsilon(\theta)-\chi V}(z) \chi V$ are invertible by the Neumann series argument. Hence we have
\begin{equation}
\label{eqn:newmerocont}
\begin{split}
    R_{H_\epsilon(\theta)}(z) & = R_{H_\epsilon(\theta)-\chi V}(z) ( I + \chi V R_{H_\epsilon(\theta)-\chi V}(z) )^{-1}  \\
    & = R_{H_\epsilon(\theta)-\chi V}(z)\sum_{j=0}^\infty (-1)^j ( \chi V R_{H_\epsilon(\theta)-\chi V}(z) )^j  \\
    & = R_{H_\epsilon(\theta)-\chi V}(z) \left( I - \chi V \,\sum_{j=0}^\infty (-1)^j ( R_{H_\epsilon(\theta)-\chi V}(z) \chi V )^j\, R_{H_\epsilon(\theta)-\chi V}(z) \right)  \\
    & = R_{H_\epsilon(\theta)-\chi V}(z) [ I - \chi V ( I + R_{H_\epsilon(\theta)-\chi V}(z) \chi V )^{-1} R_{H_\epsilon(\theta)-\chi V}(z) ],
\end{split}
\end{equation}

Using \eqref{eqn:RH-chiV}, we have
\[
R_{H_\epsilon(\theta)-\chi V}(z)\chi V = ( I + R_{\epsilon,\theta} ( z ) (1-\chi)V )^{-1} R_{\epsilon,\theta}(z) \chi V.
\]
For $\epsilon>0$, $R_{\epsilon,\theta}(z): L^2(\RR^n)\to H^2(\RR^n)\cap \langle\phi_\theta(x)\rangle^{-2}L^2(\RR^n)$, then Lemma \ref{lem:cpct embed} implies that $R_{\epsilon,\theta}(z): L^2(\RR^n)\to L^2(\RR^n)$ is compact. For $\epsilon=0$, note that $\chi V: L^2(\RR^n)\to \langle x\rangle^{-2} L^2(\RR^n)$, by Lemma \ref{lem:Rtheta} we have $R_\theta(z)\chi V: L^2(\RR^n)\to H^2(\RR^n)\cap\langle x\rangle^{-2} L^2(\RR^n)$, then Lemma \ref{lem:cpct embed} implies that $R_\theta(z)\chi V: L^2(\RR^n)\to L^2(\RR^n)$ is compact. Hence we can conclude that $z\mapsto R_{H_\epsilon(\theta)-\chi V}(z) \chi V$ is an analytic family of compact operators for $z\in\C_\delta,\ \ 0\leq \epsilon < \epsilon_0$, and thus $z\mapsto (  I + R_{H_\epsilon(\theta)-\chi V}(z) \chi V )^{-1}$ is a meromorphic family of operators in the same range of $z$.

Then we recall Lemma \ref{lem:resonance} and \ref{lem:eigenvalue} that $R_{H_\epsilon(\theta)}(z)$ is meromorphic in $-2a(\theta)<\arg z<3\pi/2 + 2a(\theta)$, by the identity principle of meromorphic operators, we conclude that \eqref{eqn:newmerocont} holds for all $z\in \C_\delta$ in the sense of meromorphic family of operators.

To obtain the multiplicity formula, we assume that $ z \in \Omega$, then there exists a neighborhood $ z \in U \subset \Omega$ and finite rank operators $A_j,\ 1 \leq j \leq J$ such that
\[ (  I + R_{H_\epsilon(\theta)-\chi V}(w) \chi V )^{-1} - \sum_{j=1}^J \frac{A_j}{(w-z)^j}\quad\textrm{is holomorphic in }w \in U.\]
Let $C_z \subset U$ be a positively oriented circle enclosing $ z $
and containing no poles of $(  I + R_{H_\epsilon(\theta)-\chi V}(w) \chi V )^{-1}$ other than possibly $ z $, thus it also contains no poles of $( w - H_\epsilon(\theta) )^{-1}$ other than possibly $ z $ as a consequence of \eqref{eqn:newmerocont}. On the one hand, we can compute
\begin{equation}
\label{eqn:lmultproof1}
\begin{split}
    m_{\epsilon,\theta} ( z ) & = \frac{ 1}{ 2 \pi i}\tr\int_{C_z}   (  I + R_{H_\epsilon(\theta)-\chi V}(w) \chi V )^{-1} \partial_w (R_{H_\epsilon(\theta)-\chi V}(w) \chi V ) dw \\
    & = \frac{ 1}{ 2 \pi i}\tr\int_{C_z}   ( I + R_{H_\epsilon(\theta)-\chi V}(w) \chi V )^{-1} R_{H_\epsilon(\theta)-\chi V}(w)^2 \chi V  dw \\
    & = \frac{ 1}{ 2 \pi i}\tr\int_{C_z} \sum_{j=1}^J \frac{A_j R_{H_\epsilon(\theta)-\chi V}(w)^2 \chi V}{(w-z)^j} dw \\
    & = \sum_{j=1}^J \frac{1}{(j-1)!} \tr \partial_z^{j-1} ( A_j R_{H_\epsilon(\theta)-\chi V}(z)^2 \chi V ) \\
    & = \sum_{j=1}^J \sum_{k=0}^{j-1} \frac{1}{k!(j-1-k)!} \tr A_j \, \partial_z^k R_{H_\epsilon(\theta)-\chi V}(z)\,
    \partial_z^{j-1-k} R_{H_\epsilon(\theta)-\chi V}(z) \chi V.
\end{split}
\end{equation}
On the other hand, by \eqref{eqn:newmerocont} we have
\begin{equation}
\label{eqn:lmultproof2}
\begin{split}
    { }&\quad\  \frac{1}{ 2 \pi i } \tr \oint_z ( w - H_\epsilon(\theta) )^{-1} dw  \\
    & = \frac{1}{ 2 \pi i } \tr \int_{C_z} R_{H_\epsilon(\theta)-\chi V}(w) \chi V ( I + R_{H_\epsilon(\theta)-\chi V}(w) \chi V )^{-1} R_{H_\epsilon(\theta)-\chi V}(w) dw \\
    & = \frac{1}{ 2 \pi i } \tr \int_{C_z} \sum_{j=1}^J \frac{R_{H_\epsilon(\theta)-\chi V}(w) \chi V A_j R_{H_\epsilon(\theta)-\chi V}(w)}{(w-z)^j} dw \\
    & = \sum_{j=1}^J \frac{1}{(j-1)!} \tr \partial_z^{j-1} ( R_{H_\epsilon(\theta)-\chi V}(z) \chi V A_j R_{H_\epsilon(\theta)-\chi V}(z) ) \\
    & = \sum_{j=1}^J \sum_{k=0}^{j-1} \frac{1}{k!(j-1-k)!} \tr \partial_z^{j-1-k} R_{H_\epsilon(\theta)-\chi V}(z) \chi V A_j \, \partial_z^k R_{H_\epsilon(\theta)-\chi V}(z).
\end{split}
\end{equation}
Now we compare \eqref{eqn:lmultproof1} and \eqref{eqn:lmultproof2}. Since $A_j$ factors have finite rank, we can apply cyclicity of the trace to obtain the multiplicity formula \eqref{eqn:lmult}.

\end{proof}

\section{Proof of convergence}
\label{poc}

The proof of convergence is based on Lemma \ref{lem:resonance}, Lemma \ref{lem:eigenvalue}, Lemma \ref{lem:MeroContinuation} and the following lemma:

\begin{lem}
\label{lem:poc estimate}
Fix any $ \theta\in D_{\beta_0}\cap \CC^+ $ and $\Omega\Subset \{ z: - 2a(\theta) < \arg z < 3 \pi /2 + 2a(\theta) \}$, there exists $\chi\in\CIc(\RR^n)$, $\chi\equiv 1$ on $B(0,T)$ for some $T>0$ such that for $0<\epsilon<\epsilon_0$,
\[ T_{\epsilon,\theta}(z):= ( H_\epsilon(\theta)-\chi V - z)^{-1} \phi_\theta(x)^2 ( H(\theta)-\chi V - z)^{-1}\chi V\]
is an analytic family of operators: $L^2(\RR^n)\to L^2(\RR^n)$. Furthermore, there exists $C=C_{\Omega,\theta}$ such that
\begin{equation}
\label{eqn:Teps estimate}
\|T_{\epsilon,\theta}(z)\|_{L^2(\RR^n)\to L^2(\RR^n)}\leq C,\quad z\in\Omega,\quad\textrm{uniformly for }0<\epsilon<\epsilon_0.
\end{equation}
\end{lem}

\begin{proof}
We recall the proof of Lemma \ref{lem:MeroContinuation} that for $T$ sufficiently large, $H_\epsilon(\theta)-\chi V - z$ is invertible, then \eqref{eqn:Neumann estimate} and \eqref{eqn:RH-chiV} imply that \[
\| (H_\epsilon(\theta)-\chi V - z)^{-1} \|_{L^2(\RR^n)\to L^2(\RR^n)}\leq C_\Omega,\quad z\in\Omega,
\]
for some $C_\Omega>0$. Hence it suffices to prove
\begin{equation}
\label{eqn:Teps estimate2}
\| \phi_\theta(x)^2 ( H(\theta)-\chi V - z)^{-1}\chi V\|_{L^2(\RR^n)\to L^2(\RR^n)}\leq C_\Omega,\quad z\in\Omega.
\end{equation}
By Lemma \ref{lem:Rtheta}, we have $\|R_\theta(z)\|_{\langle x\rangle^{-2}L^2 \to\langle x\rangle^{-2}L^2 }\leq C$. We can choose $T$ sufficiently large such that \eqref{eqn:Neumann estimate} still holds and $\|(1-\chi) V\|_{L^\infty}<1/2C$, then we have
\[\|R_\theta(z)(1-\chi V)\|_{\langle x\rangle^{-2}L^2 \to\langle x\rangle^{-2}L^2 } < 1/2,\quad z\in\Omega.\]
Hence $( I + R_\theta(z)(1-\chi) V )^{-1}:L^2\to L^2$ defined by the Neumann series in the proof of Lemma \ref{lem:MeroContinuation} also maps $\langle x\rangle^{-2}L^2$ to $\langle x\rangle^{-2}L^2$ by the same Neumann series and we have
\begin{equation}
\label{eqn:Teps estimate3}
    \|( I + R_\theta(z)(1-\chi) V )^{-1}\|_{\langle x\rangle^{-2}L^2\to \langle x\rangle^{-2}L^2}<2,\quad z\in\Omega.
\end{equation}
Since $\chi V:L^2\to\langle x\rangle^{-2}L^2 $ with the operator norm bounded by $\|\langle x\rangle^2 \chi V\|_{L^\infty}= C_\Omega$, by Lemma \ref{lem:Rtheta}, \eqref{eqn:RH-chiV} and \eqref{eqn:Teps estimate3} we conclude that
\[
\begin{split}
    { }&\quad\  \|(H(\theta)-\chi V -z )^{-1}\chi V\|_{L^2\to \langle x\rangle^{-2}L^2}  \\
    & = \|( I + R_\theta(z)(1-\chi) V )^{-1} R_\theta(z)\chi V\|_{L^2\to \langle x\rangle^{-2}L^2} \\
    & \leq \|( I + R_\theta(z)(1-\chi) V )^{-1}R_\theta(z)\|_{\langle x\rangle^{-2}L^2\to \langle x\rangle^{-2}L^2} \|\chi V\|_{L^2\to \langle x\rangle^{-2}L^2} \\
    & \leq C_\Omega,
\end{split}
\]
which implies \eqref{eqn:Teps estimate2}.
\end{proof}

Now we state the result about the convergence of eigenvalues of the deformed operator $H_\epsilon(\theta)$:

\begin{thm}
\label{t:2}
Fix $ \theta \in D_{\beta_0}\cap\CC^+ $ and $\Omega\Subset \{ z: - 2a(\theta) < \arg z < 3 \pi /2 + 2a(\theta) \}$, there exists  $\delta_0=\delta_0(\Omega)$ satisfying the following:

For any $0<\delta<\delta_0$ there exists $\epsilon'>0$ such that for any $z\in\Omega$ with $m_\theta(z)>0$ and $0<\epsilon<\epsilon'$, $H_\epsilon(\theta)$ has $m_\theta(z)$ eigenvalues in $B(z,\delta)$, where $m_\theta(z)$ is the multiplicity of the eigenvalue of $H(\theta)$ at z - see \eqref{eqn:resmult}.
\end{thm}

\begin{proof}
Since the eigenvalues of $H(\theta)$ are isolated and $\overline{\Omega}$ is compact, there are finite many $z\in \Omega$ with $m_\theta(z)>0$, we denote them by $z_1,\ldots,z_J$. Then we can choose $\delta_0$ such that $B(z_j,\delta_0)$, $j=1,\ldots,J$ are disjoint.

Now we fix $\delta<\delta_0$, by Lemma \ref{lem:MeroContinuation}, $I+R_{H(\theta)-\chi V}(w)\chi V$ is invertible in $\Omega\setminus\{z_1,\ldots,z_J\}$, thus we have
\[
\| ( I + R_{H(\theta)-\chi V}(w)\chi V )^{-1} \|_{L^2\to L^2} < C(\delta),\quad w\in\partial B(z,\delta),\ \textrm{for all }z\in\{z_1,\ldots,z_J\},
\]
for some $C(\delta)>0$. We note that in the notation of Lemma \ref{lem:poc estimate},
\[ I + R_{H_\epsilon(\theta)-\chi V}(w)\chi V  - ( I + R_{H(\theta)-\chi V}(w)\chi V ) = i\epsilon T_{\epsilon,\theta}(w).\]
Hence there exists $0<\epsilon'<\epsilon_0$ such that for any $\epsilon<\epsilon'$,
\[
\| ( I + R_{H(\theta)-\chi V}(w)\chi V )^{-1} (I + R_{H_\epsilon(\theta)-\chi V}(w)\chi V  - ( I + R_{H(\theta)-\chi V}(w)\chi V ))\|<1
\]
on $\partial B(z,\delta)$. Now we apply Gohberg-Sigal-Rouch\'e theorem, see \cite{gohberg1971operator} and \cite[Appendix C.]{res} to obtain that
\[
\begin{split}
    { }&\quad\  \frac{1}{ 2 \pi i } \tr \int_{\partial B(z,\delta)} ( I + R_{H_\epsilon(\theta)-\chi V}(w) \chi V )^{-1} \partial_w(R_{H_\epsilon(\theta)-\chi V}(w)\chi V) dw  \\
    & = \frac{1}{ 2 \pi i } \tr \int_{\partial B(z,\delta)} ( I + R_{H(\theta)-\chi V}(w) \chi V )^{-1} \partial_w(R_{H(\theta)-\chi V}(w)\chi V) dw.
\end{split}
\]
Then we recall \eqref{eqn:lmult} to conclude that
\[  \frac{1}{ 2 \pi i } \tr \int_{\partial B(z,\delta)} (w-H_\epsilon(\theta))^{-1} dw = m_\theta(z), \]
which implies that $H_\epsilon(\theta)$ has $m_\theta(z)$ eigenvalues in $B(0,\delta)$.
\end{proof}

Finally, we can give the proof of Theorem \ref{t:1}.
\begin{proof}
We assume from now on that $\epsilon <\epsilon_0$. Fix any $ \Omega \Subset \{ z\,:\,-2\beta_0 < \arg z < 3\pi/2 + 2\beta_0 \}$, we can choose $ \theta \in D_{\beta_0}\cap \CC^+ $ such that
\[\Omega \Subset \{ z\,:\,-2a(\theta) < \arg z < 3\pi/2 + 2a(\theta) \}.\]
In view of Lemma \ref{lem:resonance}, we see that $\{z_j\}_{j=1}^\infty$, the resonances of $H$ in $\Omega$, can be identified as the eigenvalues of $H(\theta)$, denoted by $\{z_{\theta,j}\}_{j=1}^\infty$. Similarly, Lemma \ref{lem:eigenvalue} guarantees that $\{z_j(\epsilon)\}_{j=1}^\infty$, the eigenvalues of $H_\epsilon$ in $\{ z\,:\,-2a(\theta) < \arg z < 3\pi/2 + 2a(\theta) \}$, are the eigenvalues of $H_\epsilon(\theta)$, denoted by $\{z_{\theta,j}(\epsilon)\}_{j=1}^\infty$. Hence it suffices to show
\begin{equation}
\label{eqn:convergence}
z_{\theta,j}(\epsilon)\to z_{\theta,j},\ \ \epsilon\to 0 +,\ \ \textrm{uniformly on }\Omega,
\end{equation}
which is a direct result of Theorem \ref{t:2}.
\end{proof}

\def\arXiv#1{\href{http://arxiv.org/abs/#1}{arXiv:#1}}

\end{document}